\newif\ifabstract
\abstractfalse
\newif\iffull
\ifabstract \fullfalse \else \fulltrue \fi

\ifabstract
\documentclass[nolineno]{socg-lipics-v2021}
\else
\documentclass[11pt,a4paper]{article}
\usepackage[margin=1in]{geometry}
\fi

\usepackage{color,ifpdf,latexsym,graphicx,url}
\usepackage{colortbl,hhline}
\usepackage{amsmath,amssymb,amsthm}
\usepackage{gensymb}
\usepackage{xcolor}
\usepackage{multirow}
\usepackage{caption}
\usepackage{algorithm}
\usepackage{algpseudocode}
\usepackage{overpic}

\iffull
\fi

\usepackage{hyperref}
\hypersetup{breaklinks,bookmarksnumbered,bookmarksopen,bookmarksopenlevel=2}
{\makeatletter \hypersetup{pdftitle={\@title}}}
{\makeatletter
 \gdef\xxxmark{%
   \expandafter\ifx\csname @mpargs\endcsname\relax % in minipage?
     \expandafter\ifx\csname @captype\endcsname\relax % in figure/caption?
       \marginpar{xxx}% not in a caption or minipage, can use marginpar
     \else
       xxx % notice trailing space
     \fi
   \else
     xxx % notice trailing space
   \fi}
 \gdef\xxx{\@ifnextchar[\xxx@lab\xxx@nolab}
 \long\gdef\xxx@lab[#1]#2{\textbf{[\xxxmark #2 ---{\sc #1}]}}
 \long\gdef\xxx@nolab#1{\textbf{[\xxxmark #1]}}
 \long\gdef\xxx@lab[#1]#2{}\long\gdef\xxx@nolab#1{}%
}

{\makeatletter \gdef\fps@figure{!htbp}}

\let\realbfseries=\bfseries
\def\bfseries{\realbfseries\boldmath}

\iffull
\newtheorem{theorem}{Theorem}
\newtheorem{lemma}[theorem]{Lemma}
\newtheorem{corollary}[theorem]{Corollary}

\fi

\let\epsilon=\varepsilon
\def\defn#1{\textbf{\textit{\boldmath #1}}}

\definecolor{blue}{rgb}{0.2,0.32,0.97}

\newcommand{\ourAcknowledgments}[0]{
  We thank Joseph O'Rourke and the anonymous referees for helpful suggestions.
  This work grew out of an open problem session and a final project from
  the MIT class on Geometric Folding Algorithms: Linkages, Origami, Polyhedra
  (6.849) held Fall 2020.
}

\title{Flat Folding an Unassigned Single-Vertex Complex
  \\
  (Combinatorially Embedded Planar Graph with Specified Edge Lengths) without Flat Angles}

\iffull

\author{
  Lily Chung\thanks{Massachusetts Institute of Technology, Cambridge, MA, USA}
  \and
  Erik D. Demaine\footnotemark[1]
  \and
  Dylan Hendrickson\footnotemark[1]
  \and
  Victor Luo\footnotemark[1]
}
\date{}

\else

\author{Lily Chung}{Massachusetts Institute of Technology, Cambridge, USA}{lkdc@mit.edu}{https://orcid.org/0000-0001-7056-6155}{}
\author{Erik D. Demaine}{Massachusetts Institute of Technology, Cambridge, USA}{edemaine@mit.edu}{https://orcid.org/0000-0003-3803-5703}{}
\author{Dylan Hendrickson}{Massachusetts Institute of Technology, Cambridge, USA}{dylanhen@mit.edu}{https://orcid.org/0000-0002-9967-8799}{}
\author{Victor Luo}{Massachusetts Institute of Technology, Cambridge, USA}{vluo@mit.edu}{}{}

\authorrunning{L. Chung, E.\,D. Demaine, D. Hendrickson, V. Luo} %TODO mandatory. First: Use abbreviated first/middle names. Second (only in severe cases): Use first author plus 'et al.'

\titlerunning{Flat Folding an Unassigned Single-Vertex Complex without Flat Angles}

\Copyright{Lily Chung, Erik D. Demaine, Dylan Hendrickson, and Victor Luo} %TODO mandatory, please use full first names. LIPIcs license is "CC-BY";  http://creativecommons.org/licenses/by/3.0/

 %TODO mandatory; please add comma-separated list of keywords

\acknowledgements{\ourAcknowledgments}
\EventEditors{Xavier Goaoc and Michael Kerber}
\EventNoEds{2}
\EventLongTitle{38th International Symposium on Computational Geometry (SoCG 2022)}
\EventShortTitle{SoCG 2022}
\EventAcronym{SoCG}
\EventYear{2022}
\EventDate{June 7--10, 2022}
\EventLocation{Berlin, Germany}
\EventLogo{socg-logo.pdf}

\SeriesVolume{224}
\ArticleNo{29}

\def\paragraph{\subparagraph*}

\fi

\begin{document}
\maketitle

\begin{abstract}
  A foundational result in origami mathematics is
  Kawasaki and Justin's simple, efficient characterization
  of flat foldability for unassigned single-vertex crease patterns
  (where each crease can fold mountain or valley) on flat material.
  This result was later generalized to cones of material,
  where the angles glued at the single vertex may not sum to~$360^\circ$.
  Here we generalize these results to when the material forms a \emph{complex}
  (instead of a manifold), and thus the angles are glued at the single vertex
  in the structure of an arbitrary planar graph (instead of a cycle).
  Like the earlier characterizations,
  we require all creases to fold mountain or valley, not remain unfolded flat;
  otherwise, the problem is known to be NP-complete
  (weakly for flat material and strongly for complexes).
  Equivalently, we efficiently characterize which combinatorially embedded planar graphs with
  prescribed edge lengths can fold flat, when all angles must be mountain or
  valley (not unfolded flat).
  Our algorithm runs in $O(n \log^3 n)$ time, improving on the previous
  best algorithm of $O(n^2 \log n)$.
\end{abstract}

\section{Introduction}

The graph flat folding problem asks whether a given combinatorially embedded
planar graph with prescribed edge lengths can be ``folded flat'' onto a line.
More precisely,
a \defn{flat folding} is an assignment of $x$ coordinates to vertices
that respects the edge lengths, together with a partial order on the edges
(which defines the stacking order among edges with overlapping $x$ extents)
that respects the combinatorial planar embedding and avoids crossings
(edges penetrating connections between higher and lower edge endpoints,
and improperly nested edge endpoint connections)
\cite{foldeq,Connelly-Demaine-Rote-2002-infinitesimally-locked,gfalop}.%
\footnote{In \cite{foldeq}, flat foldings are called ``linear folded states''.
  Here we use ``flat foldings'' so that they match up with the corresponding
  notions in computational origami.}
Equivalently, a flat folding is a sequence or continuum of planar embeddings
that respect the combinatorial planar embedding, avoid crossings, and
converge to the correct edge lengths and to lying on a line
\cite{Abbott-Demaine-Gassend-2009,jocg}.  Figure~\ref{fig:invalid} shows an example of a flat folding, as well as some non-examples.

\begin{figure}
  \centering
  \includegraphics[width=\linewidth]{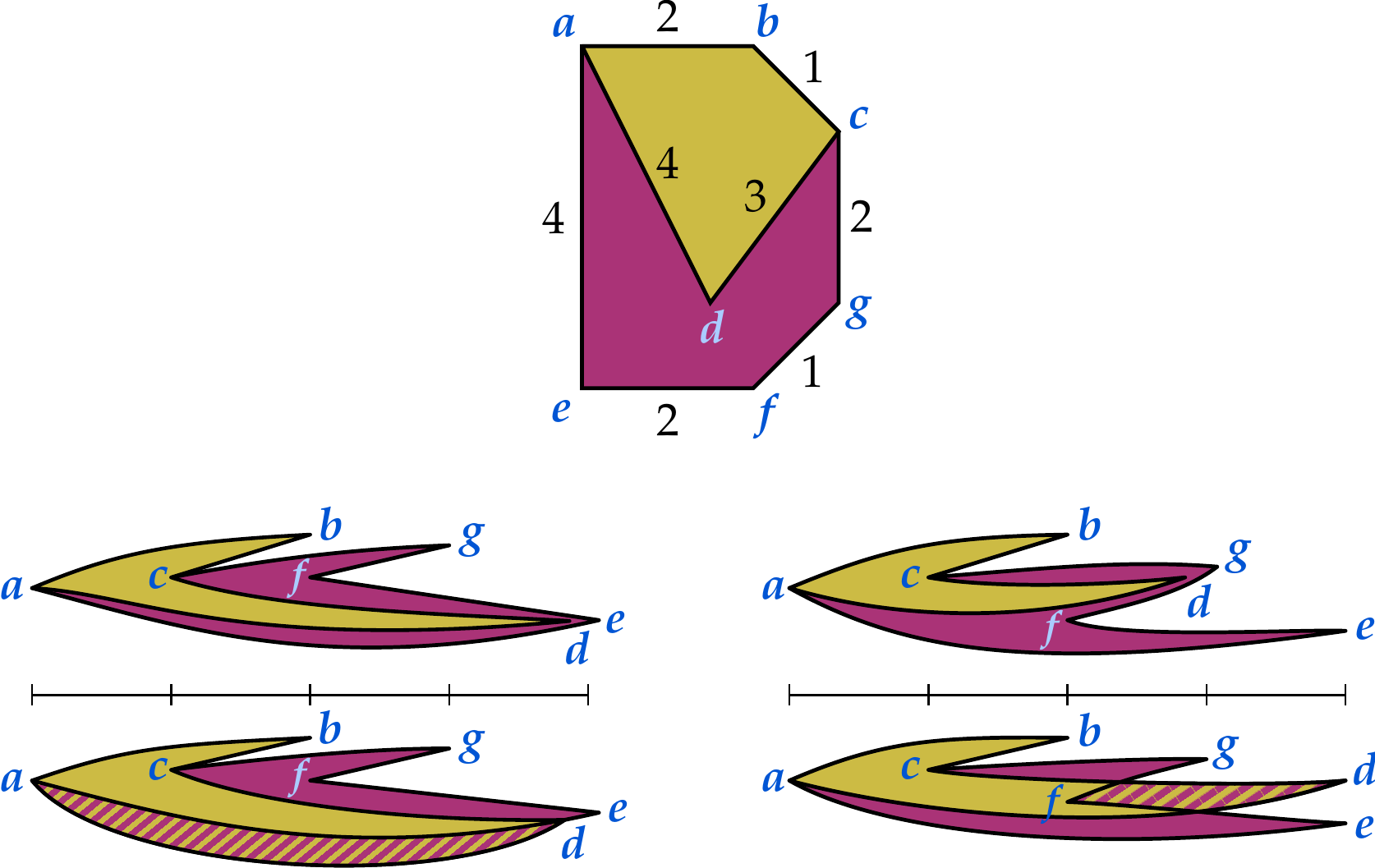}
  \caption{Four attempts to fold a graph with assigned edge lengths.  Top left: A valid folding of the graph.  Top right: Invalid because the lengths of edges \(ad\) and \(cd\) do not correspond to the original edge lengths.  Bottom left: Invalid for two reasons: the cyclic ordering of edges at vertex \(a\) is not respected, and the folding exhibits incorrect layering at vertices \(d\) and \(e\).  Bottom right: Invalid since edges cross over each other.}
  \label{fig:invalid}
\end{figure}

It is known that the graph flat folding problem is
strongly NP-complete in general,
and solvable in linear time if all edge lengths are equal
\cite{foldeq}.  But there are two natural variations on the problem,
posed in the same paper \cite{foldeq}.
In any flat folding, we can identify the angles between consecutive edges
around a vertex (as determined by the combinatorial planar embedding)
as either \defn{valley} ($0^\circ$), \defn{mountain} ($360^\circ$), or
\defn{unfolded/flat} ($180^\circ$).
(At each vertex, these angles must sum to $360^\circ$, so there is either
one mountain or two flats, and the rest are valleys.)
Now we can vary two aspects of the problem:
\begin{enumerate}
\item What if we are also given the angle (valley/mountain/flat) between
      every consecutive pair of edges around each vertex?
\item What if we forbid flat angles, and
      instead require just valleys and mountains?
\end{enumerate}

These parameters define four versions of the problem, as summarized in
Table~\ref{table:graph folding}.
The original paper \cite{foldeq} proved NP-completeness of the version
with no angles given and allowing flat angles.
Recent work shows that, if the angles are given, the problem becomes solvable
in linear time (independent of whether flat angles are allowed)
\cite{jocg}.
The remaining problem, studied here,
is the version where the angles are not given,
but flat angles are forbidden.

\begin{table}
  \definecolor{header}{rgb}{0.29,0,0.51}
  \definecolor{hard}{rgb}{1,0.85,0.85}
  \definecolor{open}{rgb}{0.95,0.95,0.5}
  \definecolor{easy}{rgb}{0.85,0.85,1}
  \newcommand\header{\cellcolor{header}\color{white}}

	\centering
	\begin{tabular}{cc|c|}
    \arrayrulecolor{white}
		& \header Flat angles forbidden & \header Flat angles allowed \\\hline
    \arrayrulecolor{black}
		\header Angles given & \multicolumn{2}{c|}{\cellcolor{easy}Linear time \cite{jocg}} \\
    \hhline{>{\arrayrulecolor{white}}->{\arrayrulecolor{black}}-|-|}
		\header Angles unspecified & \cellcolor{easy}$O(n^2 \log n)$ \cite{weak-embedding}${}\to \boldsymbol{O(n\log^3n)}$ [new] & \cellcolor{hard}NP-complete \cite{foldeq} \\ \hline
	\end{tabular}
	\caption{Complexity of different models of graph flat folding
    (based on \cite[Table~1]{jocg}, which in turn is based on open problems
    from \cite{foldeq}). Our new result is in the bottom-left.}
	\label{table:graph folding}
\end{table}

\paragraph{Connection to weak embeddings of graphs.}
Although not stated explicitly, this no-flat-angles
graph flat folding problem can be solved in polynomial time
by a reduction to ``weak embeddings of graphs''.
A key feature of this version of graph folding (in particular distinguishing it
from the NP-complete version with unknown angles that can be flat) is that the
relative coordinates of the vertices are determined by the input:
fixing one edge to go right from the origin, any path in the graph
alternates between going right and left by the specified edge lengths,
so a depth-first search fixes the vertex coordinates
\label{closure}
(and checks geometric closure constraints on cycles in the graph).
The graph flat folding problem is then equivalent to asking whether this
mapping from vertices to coordinates is a \defn{weak embedding} of the graph,
meaning that the vertices can be perturbed in the plane
within $\epsilon$-radius disks (for any $\epsilon > 0$),
and the edges can be similarly perturbed to Jordan curves
within distance~$\epsilon$ of the corresponding line segments,
so that we obtain a strict embedding
(no intersections except as intended at shared vertices).
Recognizing weak embeddings was recently solved in $O(n^2 \log n)$ time
\cite{weak-embedding},%
\footnote{The same paper \cite{weak-embedding} develops an $O(n \log n)$
  algorithm for weak embedding of graphs, but only when the given map
  is ``simplicial'', meaning that edges do not pass through other vertices.
  This property does not hold in general in the graph flat folding problem.}
so the same result applies to no-flat-angles graph flat folding.

\paragraph{Our results.}
In this paper, we give a faster algorithm for the no-flat-angles
graph flat folding problem.
Specifically, we show how to determine whether a graph can be folded flat
without flat angles in $O(n \log^3 n)$ time, which is tight up to
logarithmic factors.

We extend this result to the case where some angles are specified as flat,
and the problem asks to determine mountain or valley for each of
the remaining angles.
(The same extension also follows from the reduction to weak embedding.)
Thus what makes graph flat folding hard is not the existence of flat angles,
but deciding which angles are flat.

\paragraph{Application to single-vertex origami.}
The version of the graph flat folding problem we study
is particularly natural when viewed from the lens of computational origami.

\begin{figure}
  \centering
  \raisebox{-0.5\height}{%
    \begin{overpic}[trim=600 0 450 0, clip, scale=0.15]{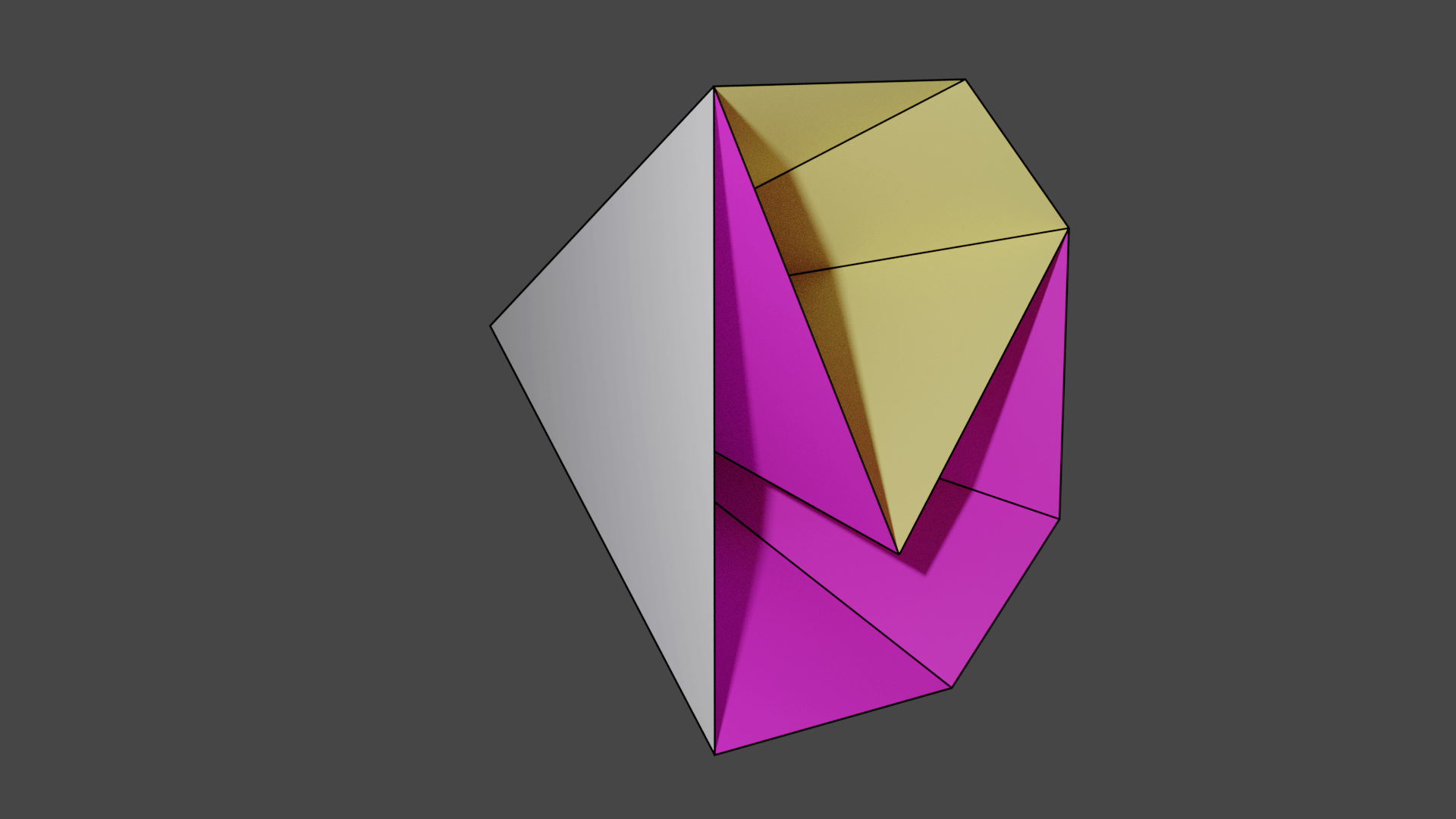}
      \put(0.5,60.5){\makebox(0,0)[l]{\color{white}$v$}}
    \end{overpic}
  }
  \hfill
  \raisebox{-0.5\height}{%
    \begin{overpic}[trim=150 225 100 300, clip, scale=0.15]{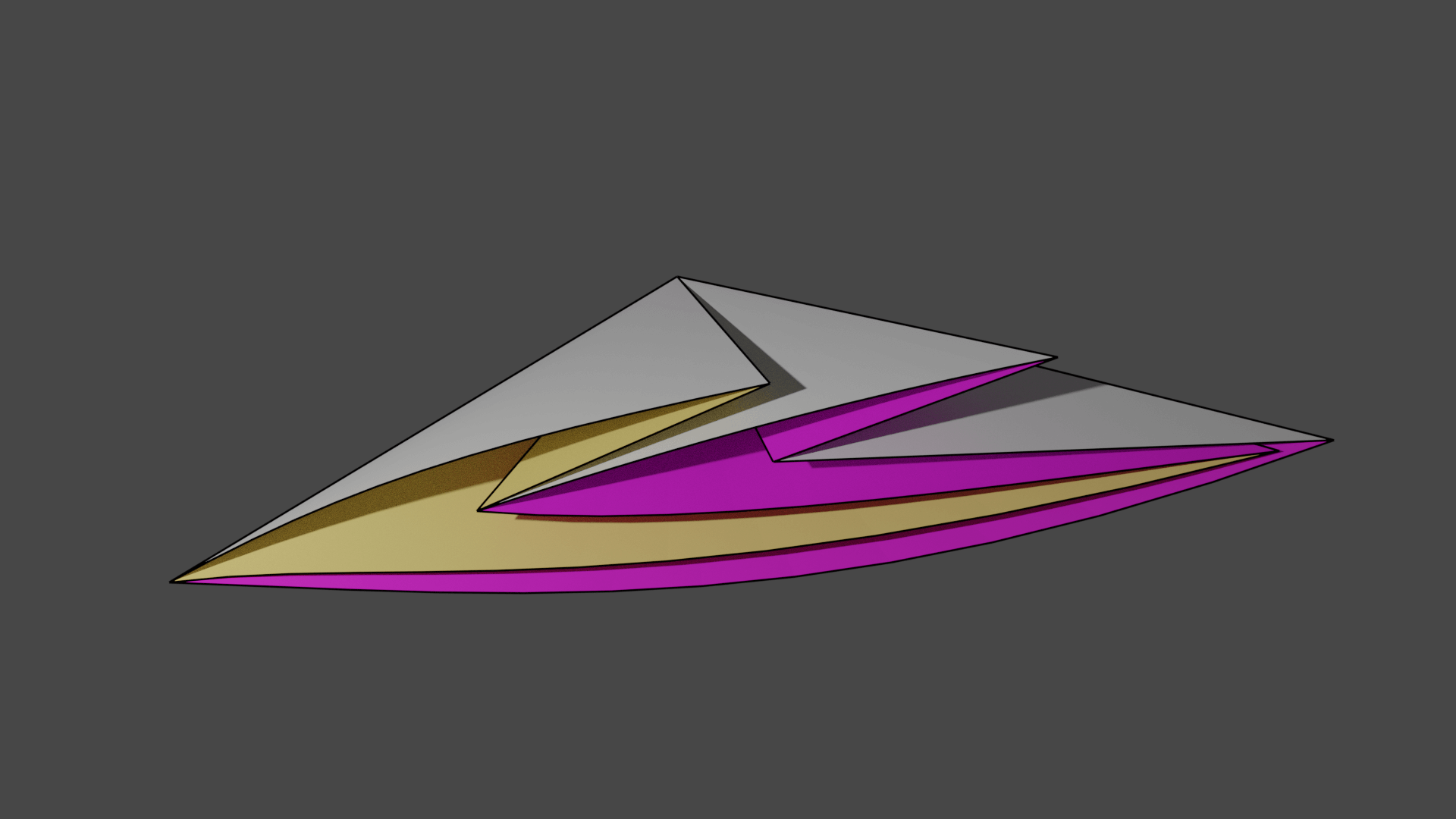}
      \put(44,31){\makebox(0,0)[c]{\color{white}$v$}}
    \end{overpic}
  }
  \caption{Unfolded and folded states of the single-vertex complex
  corresponding to the planar graph in Figure~\ref{fig:invalid}.}
  \label{fig:3d foldings}
\end{figure}

Define a \defn{single-vertex complex} to consist of $m$ polygons in 3D
where the polygons all share a common vertex~$v$, and all the shared
edges between these polygons are incident to~$v$,
as in Figure~\ref{fig:3d foldings} (left).
If we intersect such a single-vertex complex with a small sphere
centered at~$v$, we obtain a planar graph embedded on the sphere,
whose $m$ edge lengths are proportional to the $m$ polygon angles at~$v$.
In the example of Figure~\ref{fig:3d foldings} (left),
we obtain the planar graph in Figure~\ref{fig:invalid} (top).
A flat folding of the single-vertex complex into the plane
(according to standard origami definitions \cite{gfalop})
corresponds to a flat folding of the combinatorially embedded planar graph
with prescribed edge lengths \cite{foldeq,jocg}.
Figure~\ref{fig:3d foldings} (right) shows such a flat folding,
corresponding to the graph flat folding in
Figure~\ref{fig:invalid} (top left).
We can similarly consider the case of a single-vertex abstract complex ---
that is, an abstract metric space (not embedded in 3D) formed by
gluing planar polygons along edges, which all share a common vertex ---
together with the cyclic ordering of polygons around each shared edge.
Intersecting a single-vertex abstract complex with a small intrinsic sphere
centered at the shared vertex produces a graph flat folding problem,
and we can construct an arbitrary combinatorially embedded planar graph
with prescribed edge lengths by a suitable single-vertex abstract complex.
Indeed, we can construct a multigraph in this way, so we generally allow
graphs with multiple edges between the same two vertices.
Therefore graph flat folding is equivalent to origami flat foldability of
single-vertex (abstract) complexes.

When the planar graph is a cycle corresponding to $360^\circ$ of total angle
of polygons glued at a single vertex, we obtain what is known as a
\defn{single-vertex crease pattern} \cite[Section~12.2]{gfalop}:
creases emanating from a single vertex on a piece of paper.
At the first OSME (Origami Science/Mathematics/Education) conference in 1989,
Justin \cite{Justin-1989-math} and Kawasaki \cite{Kawasaki-1989a} presented
characterizations of which single-vertex crease patterns fold flat:
exactly those whose alternating sum of angles is zero.
(A complete proof of this characterization was not published until
Hull's 1994 paper \cite{Hull-1994}; see \cite[Section~5.9]{Hull-2020}.)
Crucially, this linear-time characterization
assumes that all creases must be folded either mountain or valley
(none can be left unfolded flat at an angle of $180^\circ$);
otherwise, single-vertex flat foldability
becomes weakly NP-complete~\cite{849-hinge-hardness}.

We see a similar behavior in Table~\ref{table:graph folding}
(bottom row),
where allowing mountain, valley, and flat angles makes the problem
NP-complete (even strongly), while our result shows that
allowing just mountain or valley makes the problem solvable in near-linear time.
Thus our result can be seen as a generalization of the Justin--Kawasaki
Theorem from flat paper to complexes with similar running time.
Previously, the theorem was generalized to cones of paper, where the angles
sum to a value other than $360^\circ$ \cite[Section 12.2.1]{gfalop},
but ours is the first generalization from manifolds to complexes
with near-linear running time.

The top row of Table~\ref{table:graph folding} corresponds to
single-vertex \defn{mountain-valley patterns}, where each crease is marked
as mountain or valley.  (Some creases could be marked unfolded/flat,
but this is equivalent to removing the crease.)
The previous work on given-angle complexes \cite{jocg}
can similarly be seen as a generalization of the previously known
linear-time characterization of single-vertex mountain-valley patterns
\cite{bernhayes}, \cite[Section 12.2.2]{gfalop}.

\paragraph{Organization.}
The rest of this paper is organized as follows.
First we restate two needed previous results in Section~\ref{sec:prior results}.
We then give a high-level overview of our algorithm in
Section~\ref{sec:description}, and detail the various components of the algorithm in Sections \ref{sec:face constraints}, \ref{sec:vertex constraints},
\ref{sec:solving csp}, and~\ref{sec:putting-together}.

\section{Background}\label{sec:prior results}

Our results rely on two previous results, which we restate here for completeness.

First, based on results of Hull~\cite{Hull-2001-survey},
Demaine and O'Rourke \cite{gfalop} characterized the flat-foldable mountain/valley assignments of a cycle, which we will apply to each face in a connected combinatorially embedded planar graph:

\begin{lemma}[{\cite[Corollary 12.2.12]{gfalop}}]
  \label{lem:cycle}
  Let \(f\) be a simple cycle with edge lengths \(\theta_1, \dots, \theta_n\).  
  If the edge lengths are all equal, then a crease assignment on \(f\) is flat foldable in precisely the following cases:
  \begin{itemize}
  \item \emph{Case A}: The cycle \(f\) is an interior face with equal-length edges, and there are exactly 2 more valley folds than mountain folds.
  \item \emph{Case B}: The cycle \(f\) is an exterior face with equal-length edges, and there are exactly 2 more mountain folds than valley folds.
  \end{itemize}

  Otherwise, take any maximal sequence \(e_m, \dots, e_{m + k - 1}\) of \(k\) contiguous equal-length edges surrounded by strictly longer edges, so that\footnote{The indices should be understood as being modulo \(n\).}
  \[\theta_{m - 1} > \theta_m = \dots = \theta_{m + k - 1} < \theta_{m + k}\]
  Then a crease assignment is flat foldable in precisely the following cases:
  \begin{itemize}
  \item \emph{Case C}: \(k\) is odd, and there are an equal number of mountain and valley folds incident to edges \(e_m, \dots, e_{m + k - 1}\).  Additionally, replacing all of the edges \(e_{m - 1}, \dots, e_{m + k}\) with a single edge of length \(\theta_{m - 1} - \theta_m + \theta_{m + k}\) yields a flat-foldable face with the same crease assignment.
  \item \emph{Case D}: \(k\) is even, and the numbers of mountain and valley folds incident to edges \(e_m, \dots, e_{m + k - 1}\) differ by \(\pm 1\).  Additionally, replacing all of the edges \(e_m, \dots, e_{m + k - 1}\) with a single new vertex yields a flat-foldable face,
    where the crease assignment is the same except that it assigns the new vertex to be the same type as the majority of the folds incident to \(e_m, \dots, e_{m + k - 1}\).  (That is, the new vertex is a mountain fold in this assignment if the number of mountain folds was 1 greater than the number of valley folds.)
  \end{itemize}
\end{lemma}

Second, Abel et al.~\cite{jocg} proved that a flat folding of a graph is equivalent to a compatible folding of each face:%
\footnote{A similar style of result
  (``faces being valid implies global validity'')
  was obtained in the context of upward drawings of graphs
  \cite[Theorem~3]{bertolazzi1994upward}.
  It also does not allow flat angles.
  That result, however, does not deal with prescribed edge lengths,
  which significantly complicates whether faces are flat foldable.}

\begin{theorem}[{\cite[Theorem 2]{jocg}}]
  \label{thm:fullgraph}
  Let \(G\) be a connected multigraph with an assignment of measures to every angle in \(G\).  That is, for each angle \(a\) we are given its measure \(m_a \in \{0\degree, 180\degree, 360\degree\}\).  Suppose that, for every face \(f\), the restriction of this assignment to \(f\) yields a flat-foldable mountain-valley assignment when \(f\) is treated as a simple cycle.
  Suppose also that the assignment is \defn{compatible} in that the sum of angles around each vertex is equal to 360\degree.
  Then there exists a flat folding of \(G\) whose angles have the assigned measures.
\end{theorem}

Figure~\ref{fig:approach} shows an example of combining compatible flat foldings of individual faces to obtain a flat folding of the entire graph.

\begin{figure}
  \centering
  \includegraphics[width=\linewidth]{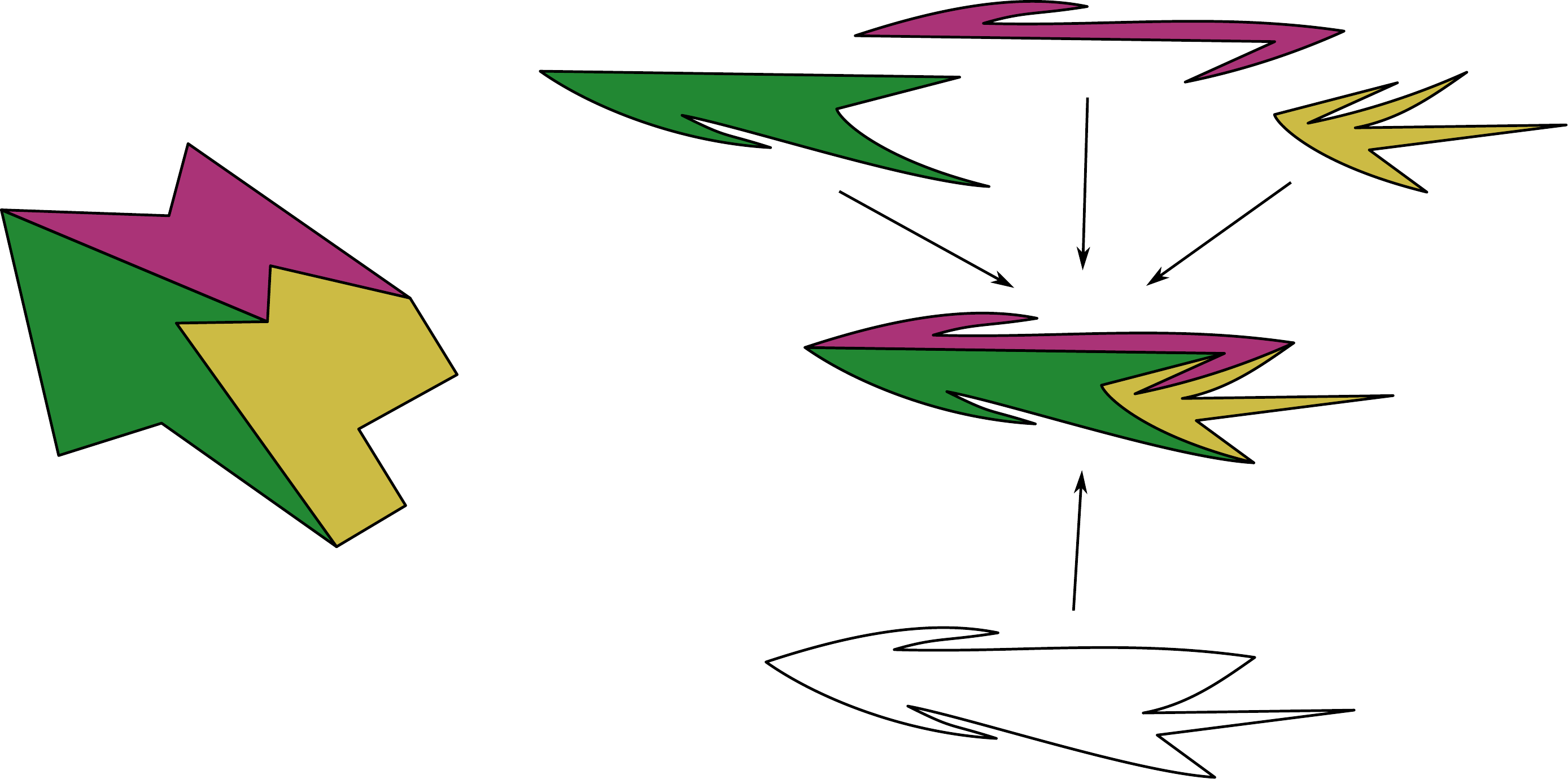}
  \caption{Left: A graph with three interior faces. Right: Given compatible flat foldings of all four faces, a flat folding of the graph can be generated.}
  \label{fig:approach}
\end{figure}

\section{Algorithm Overview}\label{sec:description}

In this section, we provide a high-level outline of our algorithm for determining whether a connected combinatorially embedded planar (multi)graph with prescribed edge lengths can be folded flat.
This algorithm takes as input a combinatorial embedding of the graph \(G\)
(which we allow to have multiple edges between the same two vertices)
and an assignment of lengths to the edges of \(G\).
We assume for now that the graph is connected and that we are given
a single face of \(G\) designated as the \defn{exterior};
in Section~\ref{sec:extensions}, we will remove both constraints.

By Theorem~\ref{thm:fullgraph}, determining whether such a graph has a flat folding is equivalent to determining whether there are compatible flat-foldable crease assignments for each face. 
To accomplish this, we reduce the graph flat folding problem to a boolean constraint satisfaction problem, with constraints deriving both from the requirement that each face needs to be flat foldable and from the compatibility requirement between faces. The variables will correspond to angles in each face of the graph (including some angles only present in virtual intermediate states), and indicate whether that angle is a valley fold or a mountain fold. The resulting constraint satisfaction problem has the following structure:
\begin{itemize}
  \item Each clause specifies an exact number of true variables in some set.
  \item Each variable appears in exactly two clauses.
  \item The graph whose vertices are clauses and whose edges are variables, connecting the two clauses in which each variable appears, is bipartite. In particular, each vertex appears in exactly one clause on each side of the bipartition.
  \item The same graph is planar.
\end{itemize}
We call such a problem \defn{planar bipartite positive $*$-in-$*$SAT-E2}.
This terminology generalizes the standard notion of
``positive $i$-in-$k$SAT'' \cite{ivan-thesis,Mulzer-Rote-2008}
where every clause requires satisfying exactly $i$ out of (up to) $k$
variables, which are never negated (hence ``positive''),
to the situation where number of variables and required true variables
may vary from clause to clause.
The standard suffix ``-E2'' represents the requirement that every variable
appears in exactly two clauses \cite{ivan-thesis,Chlebik-Chlebikova-2008}.
The ``planar'' prefix is also standard \cite{ivan-thesis,Mulzer-Rote-2008},
while the ``bipartite'' prefix is new
(and makes sense only with the ``-E2'' requirement).

In Section~\ref{sec:face constraints}, we describe the constraints which express that each face must be folded flat.
In Section~\ref{sec:vertex constraints}, we describe the constraints capturing compatibility between faces, and prove that the resulting constraint problem is equivalent to flat folding the graph.
In Section~\ref{sec:solving csp}, we show that planar bipartite positive \textsc{$*$-in-$*$SAT-E2} can be solved in $O(n\log^3n)$ time through a reduction to a flow problem.
In Section~\ref{sec:putting-together}, we put the pieces together to obtain
our main result, and describe three extensions:
to graphs with some prescribed flat angles,
to disconnected graphs, and to graphs with unknown exterior face.

\section{Single Face Constraints}\label{sec:face constraints}

In this section, we describe the constraints obtained from the requirement that each face of the graph is folded flat.  Although faces of the graph may not be bound by simple cycles (in the case of cut vertices), there exists a simple cycle corresponding to each face.  This cycle can be constructed by enumerating the face's incident edges and angles in order, duplicating any repeated vertices or edges.  Although there may exist flat foldings of this simple cycle which do not correspond to flat foldings of the original face, Theorem~\ref{thm:fullgraph} tells us that compatible flat foldings of the corresponding simple cycles are sufficient for flat foldability of the full graph.  From here on, when we discuss flat foldability of an individual face, we will actually be referring to flat foldability of the corresponding simple cycle.

Now consider flat folding a single face \(f\).
We check (and henceforth assume) that the edge lengths satisfy the
basic closure property (mentioned in Section~\ref{closure})
that the number of edges is even and
the alternating sum of edge lengths is zero;
otherwise, flat folding is impossible.
It remains to determine flat-foldable mountain/valley assignments.

We introduce a boolean variable \(x_a\) for each angle \(a\) in \(f\). These variables represent an assignment of creases: \(x_a = 0\) if \(a\) is a valley fold (0\degree) and \(x_a = 1\) if \(a\) is a mountain fold (360\degree).  We present an algorithm which, given the edge lengths and interior/exterior assignment of \(f\) and a variable \(x_a\) assigned to each angle \(a\) in \(f\), generates a set of constraints \(C_f\) on the variables \(x_a\), possibly introducing additional variables, such that solutions to this constraint problem correspond to flat-foldable crease assignments of \(f\). The constraints are of the form ``exactly $c$ variables from a set $S$ are true,'' which we write \[\sum\limits_{x\in S}x=c.\]  Additionally, each constraint generated will be colored either \textcolor{red}{red} or \textcolor{blue}{blue}; this coloring will be used later to show that the constraint satisfaction problem is bipartite.  The algorithm essentially follows Lemma~\ref{lem:cycle}:

\begin{itemize}
  \item If all edges of \(f\) have equal length, then let \(V\) be the set of all angles in \(f\),
  and let \(b\) be \(-1\) if \(f\) is an interior face, or \(+1\) if \(f\) is an exterior face.
  Generate just the \textcolor{red}{red} constraint
  \begin{equation}
    \label{eqn:equal_case}
    \sum_{a \in V} x_a\ \textcolor{red}{=}\ \frac{|V|}{2} + b.
  \end{equation}

  \item Otherwise, not all of the edges of \(f\) have equal length.  Find a sequence \(e_m, \dots, e_{m + k - 1}\) of \(k\) consecutive equal-length edges, such that \(\theta_{m - 1} > \theta_m = \dots = \theta_{m + k - 1} < \theta_{m + k}\); this is guaranteed to exist by considering a maximal sequence of consecutive edges with minimum length.  Let \(S\) be the set of angles in \(f\) incident to \(e_m, \dots, e_{m + k - 1}\).

  \item If \(k\) is odd (i.e., \(|S|\) is even), generate the \textcolor{red}{red} constraint
  \begin{equation}
    \label{eqn:odd_case}
    \sum_{a \in S} x_a\ \textcolor{red}{=}\ \frac{|S|}{2}.
  \end{equation}
  Having done so, replace all the edges \(e_{m - 1}, \dots, e_{m + k}\) with a single edge of length \(\theta_{m - 1} - \theta_m + \theta_{m + k}\) to construct a smaller face \(f^\prime\), and recursively output the constraints in \(C_{f^\prime}\).

  \item If instead \(k\) is even (i.e., \(|S|\) is odd), introduce two fresh boolean variables \(y\) and \(z\), and generate the following \textcolor{red}{red} and \textcolor{blue}{blue} (respectively) constraints:
  \begin{align}
    \label{eqn:even_case_1}
    y + \sum_{a \in S} x_a\ &\textcolor{red}{=}\ \frac{|S| + 1}{2}; \\
    \label{eqn:even_case_2}
    y + z\ &\textcolor{blue}{=}\ 1.
  \end{align}
  Then replace all the edges \(e_m, \dots, e_{m + k - 1}\) with a single new angle whose associated variable is \(z\) to construct a smaller face \(f^\prime\), and recursively output the constraints in \(C_{f^\prime}\).
\end{itemize}

\begin{figure}
	\centering
  \includegraphics[width=\linewidth]{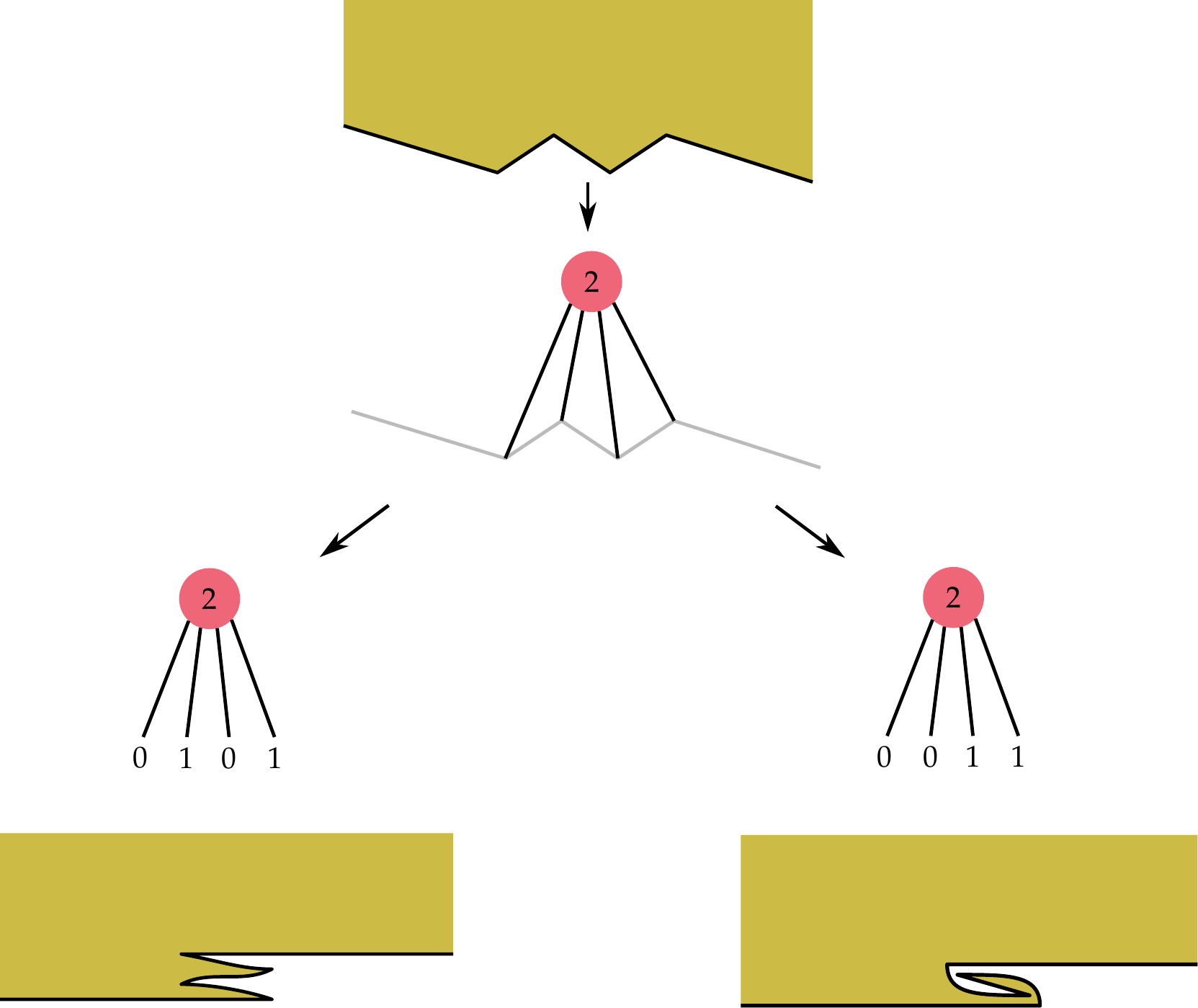}
	\caption{The case where \(k\) is odd.  In this diagram, circular vertices represent constraints and edges represent boolean variables.  Top: The \textcolor{red}{red} constraint expresses that exactly half of the creases must be mountain folds and the others must be valley folds.  Bottom: Two possible satisfying variable assignments and the associated local foldings.}
	\label{fig:odd-constraint}
\end{figure}

\begin{figure}
	\centering
  \includegraphics[height=14cm]{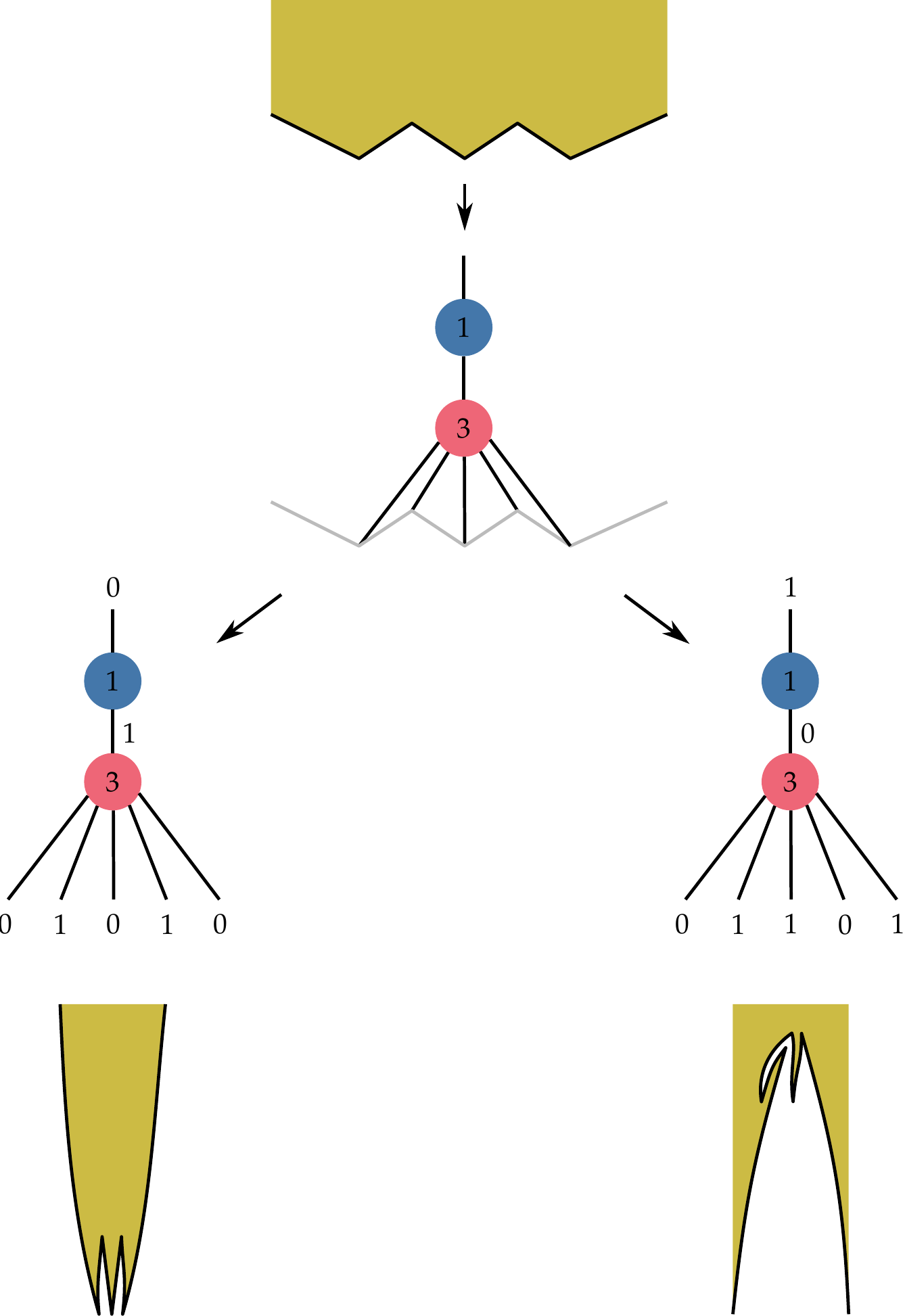}
	\caption{The case where \(k\) is even.  The pair of constraints expresses that the number of mountain folds and valley folds must differ by 1, and the majority value is equal to the newly generated variable.}
	\label{fig:even-constraint}
\end{figure}

We now show that solutions to the constraints generated by this algorithm correspond to flat-foldable crease assignments of \(f\).

\begin{theorem}\label{thm:face constraints correct}
  A mountain/valley assignment for \(f\) is flat foldable if and only if it can be extended to a satisfying assignment of \(C_f\).
\end{theorem}

We may need to extend the assignment to account for variables introduced in the case where $k$ is even. The values for these variables are forced by the constraints added when the variables are introduced, and can be determined by considering variables in the order they were added.

\begin{proof}
  The proof is by induction on the size of \(f\).
  If all the edges of \(f\) have equal length, then it is immediate by cases A and B of Lemma \ref{lem:cycle}
  that an assignment is flat foldable if and only if equation (\ref{eqn:equal_case}) holds.

  When the edges are not all equal in length, the algorithm finds some maximal sequence \(e_m, \dots, e_{m + k - 1}\) of \(k\) equal-length edges surrounded by strictly longer edges, whose incident angles we call \(S\).

  If \(k\) is odd, then by case C of Lemma \ref{lem:cycle}, the assignment is flat foldable for \(f\) if and only if it both assigns an equal number of mountain and valley folds to the angles in \(S\), and is also a flat-foldable crease assignment for \(f^\prime\), where \(f^\prime\) is the face resulting from replacing the edges \(e_{m - 1}, \dots, e_{m + k}\) with a single edge of length \(\theta_{m - 1} - \theta_m + \theta_{m + k}\).
  The first condition is just equation (\ref{eqn:odd_case}), and the second is equivalent by the inductive hypothesis to the set of constraints \(C_{f^\prime}\) obtained by recursion on \(f^\prime\). An example of this case is shown in Figure~\ref{fig:odd-constraint}.

  If \(k\) is even, then by case D of Lemma \ref{lem:cycle}, the assignment is flat foldable if and only if the mountain and valley folds assigned to the angles in \(S\) differ by 1,
  and it is also a flat-foldable crease assignment for \(f^\prime\) when suitably extended.
  Here \(f^\prime\) is the face resulting from replacing the edges \(e_m, \dots, e_{m + k - 1}\) with a single angle,
  and the assignment is extended to assign the new angle to be the same type as the majority of the folds it assigned to the angles in \(S\).
  Equation (\ref{eqn:even_case_1}) constrains the number of folds to differ by 1, where \(y\) is the minority fold type,
  and equation (\ref{eqn:even_case_2}) constrains \(z\) to be the opposite of \(y\), so \(z\) is the majority fold type.
  By the inductive hypothesis, the constraints \(C_{f^\prime}\) obtained by recursion on \(f^\prime\) are equivalent to the statement that \(f^\prime\) is flat foldable under the assignment extended to assign \(z\) to the new angle. An example of this case is shown in Figure~\ref{fig:even-constraint}.

  In all cases, we find that the assignment is flat foldable if and only if it satisfies the constraints.
\end{proof}

We also show that these constraints can be computed efficiently
and satisfy certain properties which will be useful for solving them.

\begin{theorem}\label{thm:face constraints linear}\label{thm:face constraints planar}
  The algorithm for computing \(C_f\) takes time linear in the number of angles in \(f\).  The variables and clauses of \(C_f\) form a graph in which graph vertices correspond to clauses and graph edges correspond to variables, when an additional \textcolor{blue}{blue} graph vertex is added for each angle of \(f\).  Then this graph is a bipartite (i.e. 2-colored) forest with linearly many vertices, and there is a planar embedding of this graph within \(f\) such that each vertex corresponding to an angle of \(f\) is located at the vertex of \(f\) incident to that angle.
\end{theorem}
\begin{proof}
  Let \(n\) be the number of angles in \(f\).  We prove by induction that \(C_f\) forms a graph as described with at most \(2n\) vertices.

  In the case where the edges all have equal length, \(C_f\) is a star graph whose central vertex is a \textcolor{red}{red} clause and whose outer vertices are the \textcolor{blue}{blue} angles of \(f\), so it is a bipartite forest with \(n + 1 \le 2n\) vertices.  The planar embedding can be achieved by placing the central vertex within \(f\) and drawing edges to all the vertices of \(f\).

  When the edges of \(f\) are not all equal in length, the algorithm finds some sequence of \(k\) edges whose \(k + 1\) incident angles we call \(S\).
  Let \(T\) be the star graph whose central vertex is the \textcolor{red}{red} clause added in this step and whose outer vertices are the \textcolor{blue}{blue} angles of \(S\); this is a bipartite forest with \(k + 2\) vertices.

  When \(k\) is odd, the graph \(C_f\) is simply the disjoint union of \(C_{f^\prime}\) and \(T\),
  where \(f^\prime\) is a face with \(n - k - 1\) angles.
  By the inductive hypothesis \(C_{f^\prime}\) is a bipartite forest with at most \(2(n - k - 1)\) vertices,
  so \(C_f\) is a bipartite forest with at most \(k + 2 + 2(n - k - 1) = 2n - k \le 2n\) vertices.
  The planar embedding of \(C_f\) is obtained from the planar embedding of \(C_{f^\prime}\) by simply placing \(T\) alongside it;
  none of the edges need to cross because the angles in \(S\) are contiguous in \(f\).

  When \(k\) is even, the graph \(C_f\) is formed from the disjoint union of \(C_{f^\prime}\) and \(T\) by adding an edge from the \textcolor{red}{red} central vertex of \(T\) to the \textcolor{blue}{blue} vertex corresponding to some angle \(a^\prime\) of \(f^\prime\), where \(f^\prime\) is a face with \(n - k\) angles.
  By the inductive hypothesis \(C_{f^\prime}\) is a bipartite forest with at most \(2(n - k)\) vertices,
  so \(C_f\) is a bipartite forest with at most \(k + 2 + 2(n - k) = 2n - k + 2 \le 2n\) vertices.
  The planar embedding of \(C_f\) is obtained from the planar embedding of \(C_{f^\prime}\) by first placing \(T\) alongside it as before;
  again none of the edges cross because \(S\) is contiguous in \(f\).  Then the edge from the central vertex of \(T\) to the vertex corresponding to \(a^\prime\) can be added without crossing because \(a^\prime\) occurs in the same place in \(f^\prime\)'s cyclic order of angles as \(S\) does in \(f\).

  Thus \(C_f\) is a linear-sized bipartite forest with the desired planar embedding.
  We need to show that it can be computed in linear time.
  It is straightforward to charge the work performed by the algorithm at each step to the newly created vertices,
  except for finding the sequence \(e_m, \dots, e_{m + k - 1}\) of equal-length edges surrounded by strictly longer edges.
  We cannot accomplish this by simply scanning through the edges of the face at each iteration, since this would take linear time and there might be linearly many iterations.
  We instead solve this by maintaining a cyclic doubly-linked list \(C\), each of whose entries corresponds to a maximal contiguous sequence of equal-length edges.
  Additionally we keep a list \(M\) of such entries of \(C\) which are surrounded by longer entries.
  These can be computed once at the beginning of the algorithm in linear time, and then maintained at each iteration.
  At each iteration a sequence \(e_m, \dots, e_{m + k - 1}\) is obtained by taking the first entry from \(M\) and removing it from both \(M\) and \(C\).
  When the new face \(f^\prime\) is computed,
  we add any new edges to \(C\) and check whether any of the newly adjacent pairs of entries have equal length;
  if so we consolidate them into a single entry of \(C\).
  We also check whether any of the newly adjacent entries have become surrounded by strictly longer entries; if so we add them to \(M\).
  These checks take constant time in each iteration since at most two new pairs of adjacent entries can be created.
  So computing \(C_f\) takes linear time overall.
\end{proof}

\section{Compatibility Constraints}\label{sec:vertex constraints}

Next, we describe the constraints needed to ensure that the crease assignments are compatible between faces. The angles around each vertex must sum to 360\degree; this means exactly one of these angles is a mountain fold, as shown in Figure~\ref{fig:vertex-constraint}. So for each vertex \(v\) of the graph, we generate a \textcolor{blue}{blue} constraint \(C_v\):
\begin{equation}
  \sum_{a \in A_v} x_a\ \textcolor{blue}{=}\ 1,
\end{equation}
where \(A_v\) is the set of angles incident to \(v\).

\begin{figure}
  \centering
  \includegraphics[height=12cm]{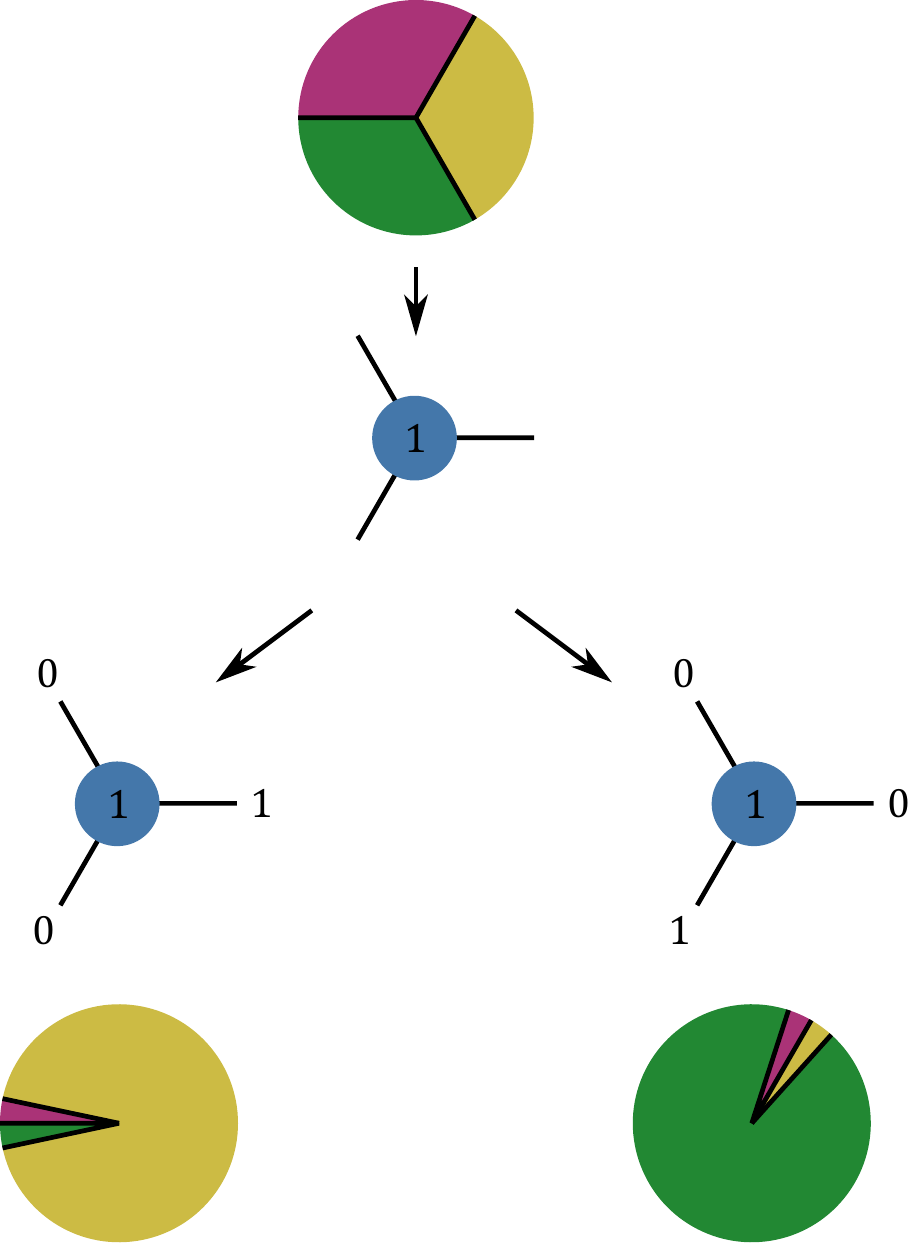}
  \caption{A vertex folds flat under a given crease assignment if and only if exactly one if the incident angles is a mountain fold.}
  \label{fig:vertex-constraint}
\end{figure}

\begin{theorem}\label{thm:constraints correct}
  A connected combinatorially embedded planar multigraph with prescribed edge lengths has a flat folding with no flat angles if and only if the constraint satisfaction problem consisting of
  \begin{itemize}
    \item for each face \(f\), the constraints \(C_f\) described in Section~\ref{sec:face constraints}, and
    \item for each vertex \(v\), the constraint \(C_v\) described above
  \end{itemize}
  is satisfiable. Moreover, these constraints can be computed in time linear in the number of angles in the graph.
\end{theorem}
\begin{proof}
  Suppose the graph has such a flat folding, and assign variables representing angles in the graph based on whether the corresponding angle is a mountain or a valley fold in the flat-folded state; this is only a partial assignment since some variables do not correspond to angles of the original graph.  Each face (and thus its corresponding simple cycle) is folded flat, and the variables which are not yet assigned are disjoint between faces, so by Theorem~\ref{thm:face constraints correct} we can extend the assignment to an assignment of all variables which satisfies \(C_f\) for every face \(f\). The assignment also satisfies \(C_v\) since exactly one angle incident to \(v\) has measure \(360\degree\) in the folded state.

  Conversely, suppose there is a satisfying assignment. Then assign each angle to be mountain or valley based on the value of the corresponding variable. By Theorem~\ref{thm:face constraints correct}, this gives a flat-foldable crease assignment for each face. These crease assignments are compatible because the variable assignments satisfy each \(C_v\), so by Theorem~\ref{thm:fullgraph} there is a flat folding with these angle assignments.

  Finally, we show that the set of constraints can be computed in linear time.
  By Theorem~\ref{thm:face constraints linear} each set of face constraints \(C_f\) can be computed in time linear in the number of angles incident to \(f\).  Since the sets of angles incident to different faces are disjoint, it takes linear time overall to compute the face constraints.
  Similarly, computing each vertex constraint \(C_v\) takes time linear in the number of angles incident to \(v\), and these are all disjoint from each other as well.
  So the set of constraints can be computed in time linear in the number of angles of the graph.
\end{proof}

\section{Solving the Constraint Satisfaction Problem}\label{sec:solving csp}

What remains is solving the constraint satisfaction problem consisting of \(C_f\) and \(C_v\) for each face and vertex of the graph.
Inspecting the constraints reveals that they are an instance of planar bipartite positive \textsc{$*$-in-$*$SAT-E2}:
\begin{itemize}
   \item Each constraint has the form $\sum\limits_{x\in S} x=c$ for some set $S$ of variables and constant $c$; this is a clause saying exactly $c$ variables in $S$ are true.
   \item The \textcolor{red}{red} and \textcolor{blue}{blue} clauses provide the bipartition.  Each variable is in exactly one \textcolor{red}{red} clause and exactly one \textcolor{blue}{blue} clause.  For each angle \(a\) incident to a face \(f\) and a vertex \(v\), the variable \(x_a\) appears in one \textcolor{red}{red} clause belonging to \(C_f\) and one \textcolor{blue}{blue} clause \(C_v\).  All other variables satisfy this condition because the subgraph corresponding to each \(C_f\) is bipartite according to Theorem~\ref{thm:face constraints planar}.
   \item The graph corresponding to the constraint satisfaction problem is planar.  We can place each clause \(C_v\) at the corresponding vertex \(v\). Then for each face \(f\) we can place the graph corresponding to \(C_f\) inside \(f\); by Theorem~\ref{thm:face constraints planar} this can be done without violating planarity.  An example of the planar embedding constructed for the entire constraint satisfaction problem is shown in Figure~\ref{fig:constraint-graph}.
\end{itemize}

All that remains to be shown is that planar bipartite positive \textsc{$*$-in-$*$SAT-E2} can be solved efficiently. We now describe a fairly standard reduction to a max-flow problem, which can be solved in near-linear time.

\begin{figure}
	\centering
  \includegraphics[height=16cm]{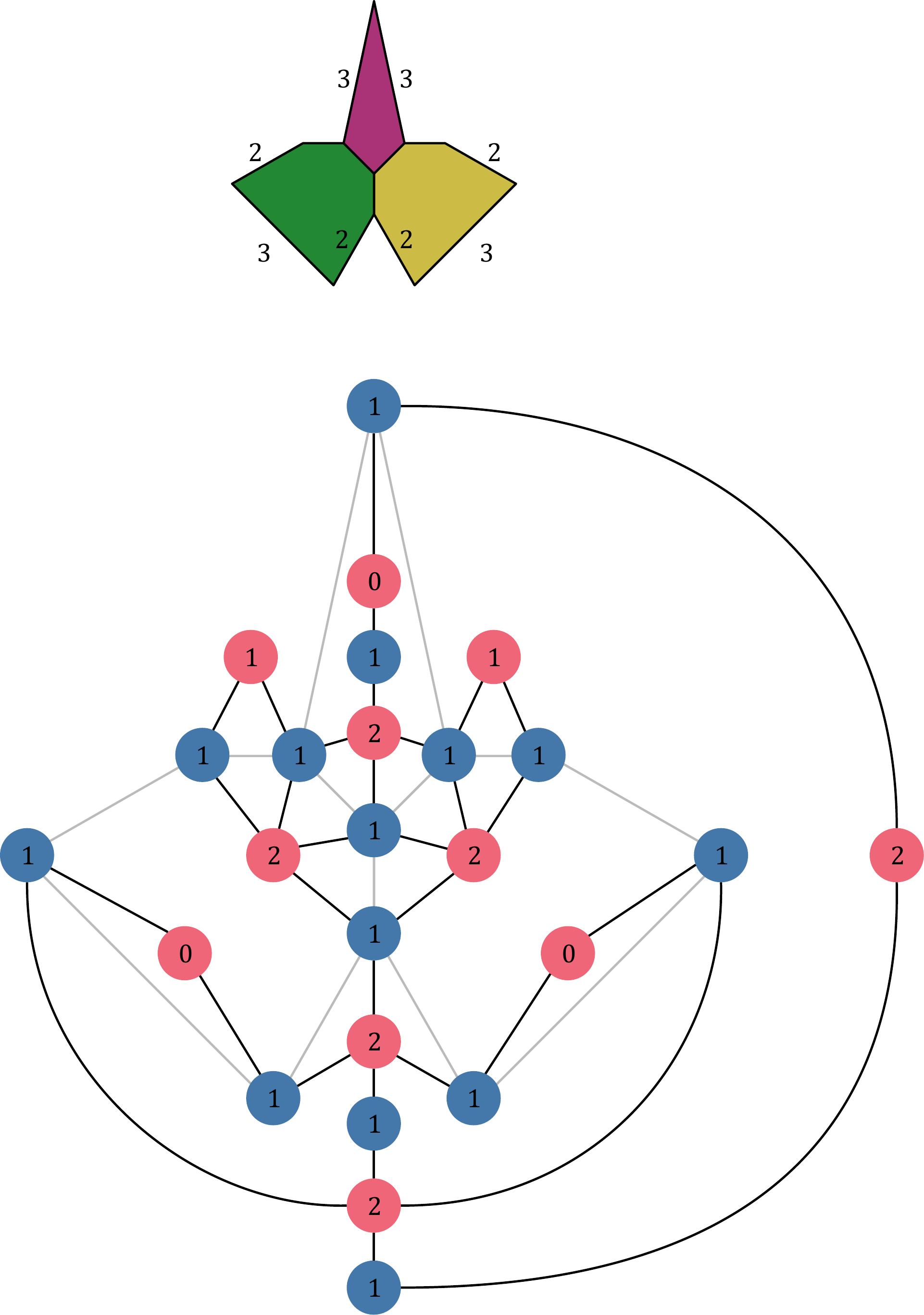}
	\caption{Top: An example graph with assigned edge lengths.  Unlabeled edges have length 1.  Bottom: The resulting instance of planar bipartite positive \textsc{$*$-in-$*$SAT-E2} (overlaid on the original graph in gray).  Since this instance is unsatisfiable, the original graph cannot be folded flat.}
	\label{fig:constraint-graph}
\end{figure}

\begin{theorem}\label{thm:solve csp}
  Planar bipartite positive \textsc{$*$-in-$*$SAT-E2} can be solved in \(O(n\log^3 n)\) time, where \(n\) is the number of clauses.
\end{theorem}

\begin{proof}
We use the graph with clauses as vertices and variables as edges, as described earlier and shown in Figure~\ref{fig:constraint-graph}.  For each \textcolor{red}{red} clause \(r\) which expects \(\ell_r\) true variables, we add a new source vertex and an edge from the source vertex to \(r\) with capacity \(\ell_r\).  Similarly, for every \textcolor{blue}{blue} clause \(b\) expecting \(\ell_b\) true variables, we add a new sink vertex and an edge from \(b\) to the sink vertex with capacity \(\ell_b\).  Finally, we assign a capacity of 1 to each edge corresponding to a variable, which goes from a \textcolor{red}{red} clause to a \textcolor{blue}{blue} clause.  This gives us an instance of multi-source multi-sink planar max-flow, for which the maximum possible flow can be determined in time \(O(k \log^3 k)\) \cite{planar-flow} where \(k\) is the number of vertices in the flow graph.  Since the flow graph has exactly twice as many vertices as there were clauses, the maximum flow can be determined in time \(O(n\log^3 n)\).

We will assume that
\[T:=\sum\limits_{\text{red }r}\ell_r=\sum\limits_{\text{blue }b}\ell_b,\]
since this is clearly required for the constraint problem to be satisfiable.

To solve the constraint satisfaction problem, we ask if the maximum flow has value $T$; this is clearly an upper bound an the maximum flow.

An integer flow uses some set of edges corresponding to variables, which specifies an assignment. The flow constraint on the edges to the appropriate source or sink forces the flow to use at most $\ell_c$ variables in clause $c$, and in order to reach the target flow $T$ we must use exactly this many variables in each clause. Thus the desired flow exists if and only if the instance of planar bipartite positive \textsc{$*$-in-$*$SAT-E2} is solvable.
\end{proof}

\section{Putting Things Together}
\label{sec:putting-together}

Combining Theorem~\ref{thm:constraints correct} and Theorem~\ref{thm:solve csp} immediately gives our main result:

\begin{corollary}
  We can determine whether a connected combinatorially embedded planar multigraph with prescribed edge lengths and exterior face has a flat folding with no flat angles in \(O(n\log^3 n)\) time, where \(n\) is the number of angles in the graph.
\end{corollary}
\begin{proof}
  The constraint problem instance can be computed in linear time, and so it has linearly many clauses, which can thus be solved in time \(O(n\log^3 n)\).
\end{proof}

\label{sec:extensions}
This result can be extended in three ways, described next.

\subsection{Extension to Specified Flat Angles}

First, we can allow flat (180\degree) angles in the folded graph, provided the input specifies \emph{which} angles are flat, leaving the remaining angles free to be mountain or valley (but not flat).

To accomplish this, first observe that for there to be a flat folding, each vertex must have exactly zero or two flat angles. We do not create variables for flat angles, since their angle is already known. Within a face that contains a flat angle, we treat the two edges around the flat angle as a single longer edge. At a vertex \(v\) which has two flat angles, we need all other angles to be valley, so the constraint \(C_v\) is now
\begin{equation}
  \sum_{a \in A_v} x_a\ \textcolor{blue}{=}\ 0,
\end{equation}
where \(A_v\) includes only non-flat angles at \(v\). The rest of the algorithm is as before.

\subsection{Extension to Disconnected Graphs}
\label{sec:disconnected}

Second, we can account for the case where the graph is disconnected.  Here we assume that the connected components are arranged in a rooted forest (i.e., a collection of rooted trees), where each non-root component specifies which interior face of its parent it is to reside in.  This condition can arise from folding an arbitrary single-vertex complex, where some faces share the central vertex but no edges; then the structure of the complex requires a certain arrangement of components within faces.  We first check that each connected component is foldable.  If this is the case, then the only obstacle to foldability is being able to fit the folded state of each child graph \(G_i\) inside the designated face \(p_i\) of its parent.

We define the \defn{folded diameter} of a graph or face to be the maximum distance between any pair of vertices in its folded state.  Since the locations of all flat angles are specified, the relative vertex coordinates are determined and can be computed in linear time (as described in Section~\ref{closure}), so this value can be computed easily without knowledge of the folding.  It turns out that we can fit \(G_i\) inside \(p_i\) in a folding if and only if the folded diameter of \(G_i\) is at most the folded diameter of \(p_i\).  To show this, we can imagine applying cases C and D of Lemma~\ref{lem:cycle} to \(p_i\) repeatedly until all the edge lengths are equal.  Because the face transformations in those cases preserve folded diameter, it follows that the remaining edges all have length equal to the folded diameter of \(p_i\).  Thus, if the folded diameter of \(G_i\) is less than or equal to the length of one of these edges, we can place a folding of \(G_i\) along it in the folded state.  On the other hand, if the folded diameter of \(G_i\) is greater than the folded diameter of \(p_i\), then we can clearly never fit a folding of \(G_i\) inside a folding of \(p_i\).

\subsection{Finding an Exterior Face}

Third, instead of assuming that the exterior face is given, we can determine
in linear time a face that is a suitable exterior face if any face is.
Observe that the exterior face must be \defn{full-diameter}
in the sense that its folded diameter
(defined in Section~\ref{sec:disconnected} above)
equals the folded diameter of the entire graph,
because some vertex of the minimum (and maximum) coordinate
must be on the exterior face in any flat folding.
We claim that \emph{every} full-diameter face
is an equally suitable exterior face:
if there is a flat folding with any one full-diameter face as exterior face,
then there is a flat folding with any desired full-diameter face
as exterior face.
Thus, to determine whether the graph is flat foldable,
we can simply find any full-diameter face and specify it as the exterior face.

\begin{figure}
  \centering
  \includegraphics[scale=0.8]{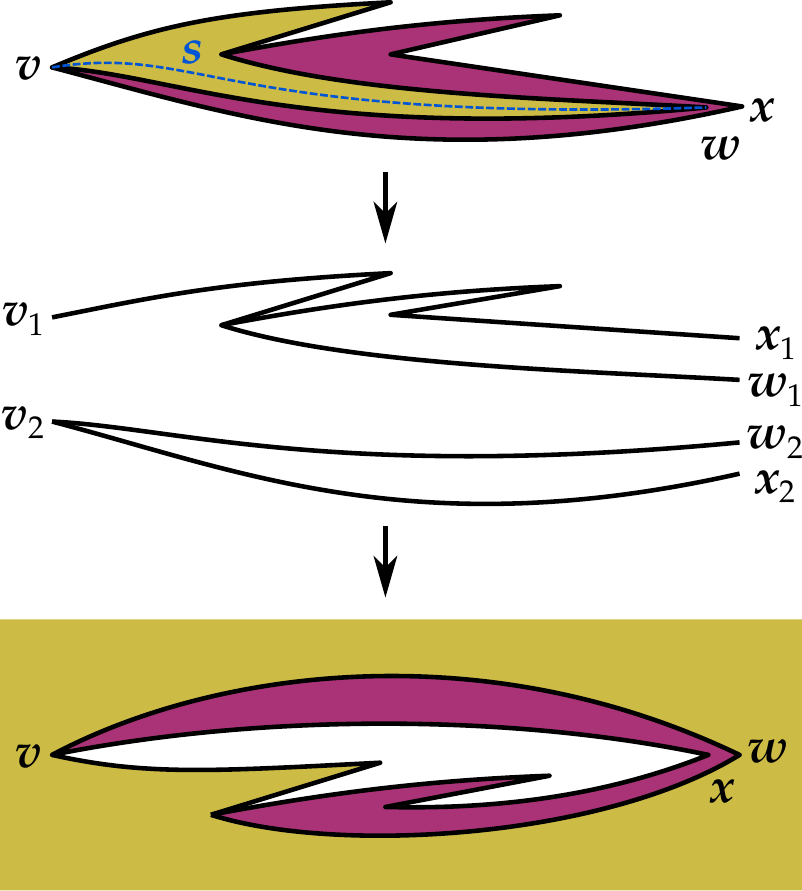}
  \caption{Cutting a flat folding apart and reassembling it with a different exterior face.}
  \label{fig:exterior rearrangement}
\end{figure}

To prove the claim, consider a flat folding of the graph,
say with exterior face~$e$.
Take any non-exterior full-diameter face~$f$,
with diameter realized by vertices $v$ and $w$.
Face $f$ consists of two folded paths connecting $v$ and~$w$.
By the argument in Section~\ref{sec:disconnected} above,
in any folding of $f$ resulting from Lemma~\ref{lem:cycle},
we can select $v$ and $w$ such that the two folded paths are separable:
we can draw a straight line segment $s$ from $v$ to $w$
that is layered in between the two folded paths.
Because $s$ is full diameter, it partitions the edges of the graph
into two halves $H_1$ and~$H_2$, where $H_1$ is entirely before $H_2$ in the layer order.
Some vertices (including $v$ and~$w$) have some incident edges
in $H_1$ and other incident edges in $H_2$.
We can imagine splitting each such vertex $x$ into two vertices $x_1$ and~$x_2$,
where $x_i$ is incident to the edges that lie within~$H_i$,
so that $H_1$ and $H_2$ become disconnected from each other.
We then swap the layer order of the two halves, placing $H_2$ before $H_1$,
and for each split vertex $x$,
reconnect the two halves $x_2$ and $x_1$,
which corresponds to a cyclic shift of the edges incident to~$x$.
This process is illustrated in Figure~\ref{fig:exterior rearrangement}.
Intuitively, we can view the folding as lying on an American football
(prolate spheroid),
where the two poles represent the minimum and maximum vertex coordinates;
then this transformation corresponds to spinning the line along which we cut
this football open to define the extremes in the other dimension
(layer order).
Thus we still obtain a flat folding of the graph,
but now $f$ is the exterior face.

\subsection{Finale}

Putting these extensions together, we have the following more general result:

\begin{corollary}
  Given a combinatorially embedded planar multigraph with prescribed edge lengths and some angles specified as flat, we can determine in \(O(n\log^3 n)\) time whether there is a flat folding that has precisely the specified angles flat.
\end{corollary}

On the other hand, if the set of flat angles is not specified, it is NP-complete to determine whether there is a flat folding \cite{foldeq}, so this implies that the hard part is deciding which angles should be flat.

\iffull
\section*{Acknowledgments}
\ourAcknowledgments
\fi

\iffull
  \bibliographystyle{alpha}
\else
  \bibliographystyle{plainurl}
\fi
\bibliography{paper}

\newcommand{\etalchar}[1]{$^{#1}$}
\begin{thebibliography}{BDBLM94}

\bibitem[ADD{\etalchar{+}}13]{foldeq}
Zachary Abel, Erik~D. Demaine, Martin~L. Demaine, Sarah Eisenstat, Jayson
  Lynch, Tao~B. Schardl, and Isaac Shapiro-Ellowitz.
\newblock Folding equilateral plane graphs.
\newblock {\em International Journal of Computational Geometry and
  Applications}, 23(2):75--92, April 2013.

\bibitem[ADD{\etalchar{+}}18]{jocg}
Zachary Abel, Erik~D. Demaine, Martin~L. Demaine, David Eppstein, Anna Lubiw,
  and Ryuhei Uehara.
\newblock Flat foldings of plane graphs with prescribed angles and edge
  lengths.
\newblock {\em Journal of Computational Geometry}, 9(1):74--93, 2018.

\bibitem[ADG09]{Abbott-Demaine-Gassend-2009}
Timothy~G. Abbott, Erik~D. Demaine, and Blaise Gassend.
\newblock A generalized carpenter's rule theorem for self-touching linkages.
\newblock arXiv:0901.1322, 2009.
\newblock \url{https://arXiv.org/abs/0901.1322}.

\bibitem[AFT19]{weak-embedding}
Hugo~A. Akitaya, Radoslav Fulek, and Csaba~D. T{\'{o}}th.
\newblock Recognizing weak embeddings of graphs.
\newblock {\em ACM Transactions on Algorithms}, 15(4), October 2019.

\bibitem[BDBLM94]{bertolazzi1994upward}
Paola Bertolazzi, Giuseppe Di~Battista, Giuseppe Liotta, and Carlo Mannino.
\newblock Upward drawings of triconnected digraphs.
\newblock {\em Algorithmica}, 12(6):476--497, 1994.

\bibitem[BH96]{bernhayes}
Marshall Bern and Barry Hayes.
\newblock The complexity of flat origami.
\newblock In {\em Proceedings of the 7th Annual ACM-SIAM Symposium on Discrete
  Algorithms}, pages 175--183, Atlanta, January 1996.

\bibitem[BKM{\etalchar{+}}17]{planar-flow}
Glencora Borradaile, Philip~N. Klein, Shay Mozes, Yahav Nussbaum, and Christian
  Wulff-Nilsen.
\newblock Multiple-source multiple-sink maximum flow in directed planar graphs
  in near-linear time.
\newblock {\em SIAM Journal on Computing}, 46(4):1280--1303, 2017.

\bibitem[CC08]{Chlebik-Chlebikova-2008}
M.~Chleb{\'{i}}k and J.~Chleb{\'{i}}kov{\'{a}}.
\newblock Approximation hardness of dominating set problems in bounded degree
  graphs.
\newblock {\em Information and Computation}, 206(11):1264--1275, 2008.

\bibitem[CDR02]{Connelly-Demaine-Rote-2002-infinitesimally-locked}
Robert Connelly, Erik~D. Demaine, and G\"unter Rote.
\newblock Infinitesimally locked self-touching linkages with applications to
  locked trees.
\newblock In J.~Calvo, K.~Millett, and E.~Rawdon, editors, {\em Physical Knots:
  Knotting, Linking, and Folding of Geometric Objects in 3-space}, pages
  287--311. American Mathematical Society, 2002.

\bibitem[Dem10]{849-hinge-hardness}
Erik~D. Demaine.
\newblock 6.849: Geometric folding algorithms: Linkages, origami, polyhedra:
  Lecture 5.
\newblock MIT class, Fall 2010.
\newblock
  \url{https://courses.csail.mit.edu/6.849/fall10/lectures/L05.html?notes=4}.

\bibitem[DO07]{gfalop}
Erik~D. Demaine and Joseph O'Rourke.
\newblock {\em Geometric Folding Algorithms: Linkages, Origami, Polyhedra}.
\newblock Cambridge University Press, 2007.

\bibitem[Fil19]{ivan-thesis}
Ivan Tadeu Ferreira~Antunes Filho.
\newblock Characterizing boolean satisfiability variants.
\newblock M.eng.~thesis, Massachusetts Institute of Technology, 2019.

\bibitem[Hul94]{Hull-1994}
Thomas Hull.
\newblock On the mathematics of flat origamis.
\newblock {\em Congressus Numerantium}, 100:215--224, 1994.

\bibitem[Hul01]{Hull-2001-survey}
Thomas Hull.
\newblock The combinatorics of flat folds: a survey.
\newblock In {\em Origami$^3$: Proceedings of the 3rd International Meeting of
  Origami Science, Math, and Education}, pages 29--38, Monterey, California,
  March 2001.

\bibitem[Hul20]{Hull-2020}
Thomas~C. Hull.
\newblock {\em Origametry: Mathematical Methods in Paper Folding}.
\newblock Cambridge University Press, December 2020.

\bibitem[Jus89]{Justin-1989-math}
Jacques Justin.
\newblock Aspects mathematiques du pliage de papier ({M}athematical aspects of
  paper folding).
\newblock In H.~Huzita, editor, {\em Proceedings of the 1st International
  Meeting of Origami Science and Technology}, pages 263--277, Ferrara, Italy,
  December 1989.
\newblock Originally appeared in \emph{L'Ouvert}, number 47, 1987, pages 1--14.

\bibitem[Kaw89]{Kawasaki-1989a}
Toshikazu Kawasaki.
\newblock On the relation between mountain-creases and valley-creases of a flat
  origami.
\newblock In H.~Huzita, editor, {\em Proceedings of the 1st International
  Meeting of Origami Science and Technology}, pages 229--237, Ferrara, Italy,
  December 1989.
\newblock An unabridged Japanese version appeared in \emph{Sasebo College of
  Technology Report}, 27:153--157, 1990.

\bibitem[MR08]{Mulzer-Rote-2008}
Wolfgang Mulzer and G{\"{u}}nter Rote.
\newblock Minimum-weight triangulation is {NP}-hard.
\newblock {\em Journal of the ACM}, 55(2), May 2008.

\end{thebibliography}

\end{document}